\documentclass[twoside,leqno,twocolumn]{article}  
\usepackage{ltexpprt} 
\usepackage{amsmath}
\usepackage{amsfonts}
\usepackage[ruled,vlined,algo2e]{algorithm2e}
\usepackage[noend]{algorithmic}

\newcommand{\B}{\mathcal{B}}

\usepackage{graphicx}
\usepackage{paralist}
\usepackage{multirow}

\usepackage{subfigure}

\SetKw{Kw}{function}

\usepackage{makeidx}

\newtheorem{lemmx}{Lemma}
\newtheorem{thmx}{Theorem}

\begin{document}

\title{\Large On Parallelizing Matrix Multiplication by the Column-Row Method\thanks{Work supported by the Danish National Research Council under the Sapere Aude program.}}
\author{Andrea Campagna\thanks{acam@itu.dk, IT University of Copenhagen, Denmark.} \\
\and 
Konstantin Kutzkov\thanks{konk@itu.dk, IT University of Copenhagen, Denmark.}\\
\and Rasmus Pagh\thanks{pagh@itu.dk, IT University of Copenhagen, Denmark.}}
\date{}

\maketitle

%\pagenumbering{arabic}
%\setcounter{page}{1}%Leave this line commented out.

\begin{abstract} \small\baselineskip=9pt We consider the problem of sparse matrix multiplication by the column row method in a distributed setting where the matrix product is not necessarily sparse. We present a surprisingly simple
method for ``consistent'' parallel processing of sparse
outer products (column-row vector products) over several processors, in a {\em communication-avoiding\/} setting where each processor has a copy of the input.
The method is consistent in
the sense that a given output entry is always assigned to the
same processor independently of the specific structure of the outer product. 
We show guarantees on the work done by each processor,
and achieve linear speedup down to the point where the cost is
dominated by reading the input.
Our method gives a way of distributing (or parallelizing) matrix product computations in settings where the main bottlenecks are storing the result matrix, and inter-processor communication. Motivated by observations on real data that often the absolute values of the entries in the product adhere to a power law, we combine our approach with frequent items mining algorithms and show how to obtain a tight approximation of the weight of the heaviest entries in the product matrix. 
\\
As a case study we present the application of our approach to frequent pair mining in transactional data streams, a problem that can be phrased in terms of sparse $\{0,1\}$-integer matrix multiplication by the column-row method. Experimental evaluation of the proposed method on real-life data supports the theoretical findings.
\end{abstract}

\section{Introduction}
% TO-DO: 
% - Address reviewer comment: "seems that each processor must also keep a copy of the input vectors in order to perform any calculation (specifically, due to an initial sort step), which mitigates any space saving of the streaming aspects. The space usage seems reasonable if the vectors are very sparse (like O(1) nonzeroes), but if they are very sparse then the processors are unlikely to be able to all be kept busy which would yield a relatively inefficient algorithm for something like matrix multiply (even ignoring the communication costs that are considered free)."
% -- The technical contribution is mostly in the right combination and application of existing techniques. The additional contribution seems limited. 

Column row vector products (aka.~outer products) are ubiquitous in scientific computing. Often one needs to compute an aggregate function over several products. 
Most notably, the product of two matrices can be written as a sum of outer products.
Observe that a product of two sparse matrices can be dense in the worst case, and will typically be much less sparse than the input matrices. This is a practical problem in algorithms that multiply very high-dimensional, sparse matrices, such as the Markov Cluster Algorithm~\cite{thesis-clust} popular in bio-informatics. There are approaches to approximating matrix products that can use less space~\cite{CohenL97,PaghITCS} but they do not address how to balance the work in a parallel or distributed setting. State-of-the-art clusters for bio-informatics can easily have a combined main memory capacity of 1 TB or more. Thus there is potential for computing huge result matrices if computation and storage can be used efficiently.

Based on the observation that handling the output data, rather than the input data, can be the main bottleneck we present a method for parallelizing the computation of sparse outer products in a setting where {\em the input is assumed to be broadcast}. No other communication is allowed --- that is, algorithms need to be {\em communication avoiding}.
To our best knowledge there is no previous work on (sparse) matrix multiplication in such a setting.
%, probably because the time for broadcsting the input has been seen as a bottleneck for massive parallelization.
The standard approach is to assume that this input is distributed among processors, which then communicate as needed.
As noted by Demmel et al.~\cite{4536305} communication is often a bottleneck in parallel sparse matrix computations.
We believe that in some settings (such as an ethernet-connected cluster) where broadcast is cheap, but the total capacity for point-to-point communication is a bottleneck, an approach that trades broadcasting the input for reduced inter-processor communication can improve performance, assuming that the total work remains the same and that the computation is load-balanced well.

Because there can be no communication after the broadcast of the input we do not need to consider a particular parallel or distributed memory model.
However, for concreteness we will from now on speak of parallelization across multiple processors.
For practical reasons we also restrict our experiments to the setting of a multi-core architecture.
Our algorithm avoids communication by distributing the matrix output entries in a {\em consistent\/} way, i.e., each processor will process either all nonzero terms contributing to a given output entry, or none of them. 
As a case study we consider 0-1 matrix products arising in data mining applications, where the task is to approximate the largest entries in the result matrix well.
This is done by combining our method for parallelizing outer products with known heavy hitter algorithms.
%We show both theoretically and empirically that by employing efficiently computable hash functions our method  is well suited for several real life applications.

Our algorithm works in a data stream setting, where the vectors are given one at a time, and space for saving past vectors is not available.
In contrast, traditional work on parallel matrix multiplication requires that the whole input can be stored in the memory of the system.
However, our main point is not the space saved on not storing the input, but rather that the (possibly non-dense) output is distributed evenly among the processors.
We believe that there are many interesting opportunities ahead in parallel processing of data streams, especially for computations that require more than (quasi-) linear time in the input size.

\section{Background} \label{background}

\paragraph{Parallel matrix multiplication.}

A major algorithmic problem that can be approached using our method is sparse matrix multiplication.
The product of matrices $A, B \in \mathbb{R}^{n\times n}$ can be computed by summing over $n$ outer products of columns of $A$ and rows of $B$. Assuming all outer products are dense a simple algorithm  distributes the computation among several processors such that the work load is equal and there is no need for communication between the processors or shared data structures: each one of $p$ processors will sum over a submatrix of the outer products of size $\frac{n}{\sqrt{p}} \times \frac{n}{\sqrt{p}}$, and only the column and row vectors have to be read by each processor. This approach, first described by Cannon~\cite{cannon}, is known to achieve very good scalability with a growing number of processors.  

However, for the situation where the outer products are expected to be sparse, this simple algorithm is not guaranteed to achieve good scalability. The reason is that we do not know in advance the specific structure of the output matrix and it might happen that certain $\frac{n}{\sqrt{p}} \times \frac{n}{\sqrt{p}}$-submatrices are dense while many others are sparse, which implies that the workload is not balanced among processors. An approach to avoid this problem is to initially permute the rows (resp.~columns) of the left (resp.~right) input matrix. This corresponds to permuting both rows and columns of the output matrix. However, this approach is vulnerable to the situation in which the nonzero output entries are not well distributed among rows or columns. For example, if half of the work in computing the output relates to a single output row or column, this work will only be distributed among $\sqrt{p}$ out of $p$ processors.

Also, it is often the case that we don't know in advance the exact dimensions of the sparse matrices as they are generated in a streaming fashion. Matrix multiplication in the streaming setting has recently received some attention~\cite{journals/siamcomp/DrineasKM06,PaghITCS,conf/focs/Sarlos06}. Randomized algorithms have been designed approximating the values of individual entries in the matrix product running in subcubic time and subquadratic space in one or two passes over the input matrices and requiring access only to certain columns and rows. 

We refer to~\cite{thesis} for an overview on state-of-the-art results on parallel matrix multiplication relying on inter-processor communication.

\section{Overview of contributions.}

\paragraph{Parallel sparse matrix multiplication.}
Hashing has long been used for load balancing tasks, mapping a given task $x$ to a random bucket $h(x)$.
A good hash function will distribute the tasks almost evenly among the buckets, such that different processors will handle given buckets and get a similar load.
When computing a matrix product we can think of each entry of the output matrix as a task.
However, if we decide to simply hash all entries in a product (in parallel), a lot of inter-processor communication will be needed to identify the entries in each bucket that are nonzero in a given outer product.
Our approach turns this around, and uses a carefully chosen hash function that allows a processor to efficiently identify the output entries hashing to a given bucket.
In particular each processor can, {\em without communication or shared data structures}, process exactly the terms that belong to its bucket, with an additional overhead that is linear in the size of the input vectors. We show that our approach is particularly well suited for matrix multiplication by the column-row method when individual outer products are sparse but not the resulting matrix product.

\paragraph{Approximate matrix multiplication.}
As already mentioned, randomized algorithms running in subcubic time, using a small amount of memory and performing only a limited number of scans of the input matrix are gaining more and more popularity~\cite{journals/siamcomp/DrineasKM06,PaghITCS,conf/focs/Sarlos06}. Instead of computing exactly the matrix product, the algorithms return an approximation of individual entries. For matrix products where the entries adhere to a skewed distribution, the approximation of the heaviest entries is known to be of very high quality. The reader is referred to~\cite{PaghITCS} for a more detailed discussion and a list of applications of approximate matrix multiplication.

Inspired by the above and observations on real data, we compose our parallel matrix multiplication method with two known algorithms for finding frequent items in data streams: {\sc Space-Saving}~\cite{spacesaving} and {\sc Count-Sketch}~\cite{cs}. The former gives deterministic upper and lower bounds on the true value and the latter -- an unbiased estimator. Both algorithms are capable of handling weighted updates but  {\sc Space-Saving} is restricted only to positive updates.

\medskip

\noindent
More concretely, our main contributions can be summarized as follows:
\begin{itemize}
\item A new algorithm for the parallelization of the multiplication of sparse matrices by the column row method.%Building upon work on extracting information from data streams we apply well-known algorithms to approximate sparse matrix multiplication. Using a suitably constructed hash function we show how to efficiently distribute the work among processors in a balanced way. Our algorithm {\bf which algorithm} does not rely on shared data structures or other communication among processors, so it is able to achieve excellent scalability. 

\item The combination of the above with approximate matrix multiplication. 
%The resulting algorithm has a user-defined space usage, and rigorously understood time complexity and output quality (the latter improving with space usage).

\item We theoretically analyze the expected time complexity, the load balancing among cores and the approximation guarantee of our algorithm under the assumption of Zipfian distribution of the entries' weights (a common situation in many real-life applications). In particular, while {\sc Count-Sketch} was always based on an initial hashing/splitting step, we believe that this work is the first to investigate the composition of hashing and {\sc Space-Saving} for approximate matrix multiplication.

\item Extensive experimental evaluation of our approach in the context of frequent pair mining in transactional data streams.
\end{itemize}

\section{Preliminaries} \label{sec:prel}

\paragraph{Notation}

For vectors $u, v \in \mathbb{R}^n$ the {\em outer product} of $u$ and $v$ is denoted as ${uv} \in \mathbb{R}^{n \times n}$. For matrices $A, B \in \mathbb{R}^{n \times n}$ we denote by $a_i$ the $i$th column of $A$ and by $b_j$ the $j$th row of $B$ for $i, j \in [n]$ where  $[n]:= \{0,\ldots,n-1\}$. The $i$th element in a vector $u$ is written as $u(i)$. The {\em weight} of the {\em entry} $(i, j)$ in the product $AB$ is the inner product of the $i$th row of $A$ and $j$th column of $B$, for $i, j \in [n]$. Alternatively, in the column-row notation it can be also written as $\sum_{k=0}^{n-1}a_k(i)b_k(j)$, i.e., the sum over the $(i, j)$th entry in all outer products. When clear from the context we refer to entries with larger absolute weight as {\em heavy} entries. The number of non-zero entries in $u \in \mathbb{R}^n$ is written as $|u|$. %The set of all entries is denoted as $\mathcal{E}$.
\\

A family $\mathcal{H}$ of functions from $\mathcal{E}$ to $[k]$ is {\em $k$-wise independent} if for a function $h: \mathcal{E} \rightarrow [k]$ chosen uniformly at random from $\mathcal{H}$ it holds $$\Pr[h(e_1) = c_1 \wedge h(e_2) = c_2 \wedge \dots \wedge h(e_t) = c_t] = 
k^{-t}$$
%
%\Pr[h(I_1) = c_1]\Pr[h(I_2) = c_2]..\Pr[h(I_t)=c_t]$ 
for distinct elements $e_i \in \mathcal{E}$, $1 \leq i \leq t$, and $c_i \in [k]$. 
We will refer to a function chosen uniformly at random from a $k$-wise independent family $\mathcal{H}$ as a $k$-wise independent function.

The elements in $\mathcal{E}$ adhere to Zipfian distribution with parameters $C$ and $z$ if $w_i = C/i^z$ for
the absolute weight $w_i$ of the $i$th most heavy element in $\mathcal{E}$.
%Note that we will consider $z$ to be a constant but $C$ does not need to be constant. 

%We will also use data mining notation throughout the paper.
%The transaction stream is denoted by $\mathcal{S} = T_1,.., T_m$ where $T_i  \subseteq [n]$.
%We call a subset $p=\{i,j\} \subset [n]$ a {\em pair}. The {\em support} of a pair $p$ is the number of transactions containing $p$: $\sup(p) = |\{T_j: p \subseteq T_j\}|, 1 \leq j \leq m$. We can consider each transaction $T$ as a $\{0,1\}$-valued sparse $n$-dimensional vector $v_T$ and the set of pairs occurring in the given transaction as the nonzero entries above the diagonal in the outer product of $v_T$ and its transpose.

%The set of pairs is denoted by $\mathcal{P}$, while the number of distinct pairs
%occurring in the stream $\mathcal{S}$ as $d \leq \binom{n}{2}$. 
%Furthermore the number of frequent pairs by $f$, where the meaning of {\em frequent}
%will be specified in the given context.

\paragraph{Frequent items mining algorithms.} \label{cs_ss_overview}

We will use two well-known frequent item mining algorithms {\sc Count-Sketch}~\cite{cs} and {\sc Space-Saving}~\cite{spacesaving} as subroutines. We give a brief overview of how they work.
 
In {\sc Count-Sketch} every item $i$ is hashed by a hash function
$h:[n] \rightarrow [k]$ to a position in an array $\mathcal{CS}$ consisting of $k$ estimators, each of them being a real number. Upon updating the weight of an item $i$, we add  $1$ or $-1$ to the corresponding estimator by using a uniform sign hash function $s(i)$ evaluating $i$ to either $1$ or $-1$.
After processing the stream the frequency of a given item $i$ can be estimated
as $\mathcal{CS}[h(i)]\cdot s(i)$.
The intuition is that for a heavy item the contribution from other items will cancel out and will not significantly affect the estimate. Both $h$ and $s$ are pairwise independent and this is sufficient to show that
for an appropriate number of estimators and a skewed distribution of item weights the heavy items will be assigned to only one estimator with high probability. The probability for correct estimates can be amplified by working with $t>1$ hash functions and returning the median of the $t$ estimates upon a query for the frequency of a given item.

The {\sc Space-Saving} algorithm offers upper and lower frequency bounds,
rather than an unbiased estimator.
It keeps a summary of the stream consisting of $\ell$ triples $(\text{item}_j, \text{count}_j, \text{overestimation}_j)$, $1 \leq j \leq \ell$. The $\ell$ triples are sorted according to their {\em count} value.  Upon the arrival of a new item $i$ the algorithm distinguishes between the following cases:
\begin{itemize}
\item If not all $\ell$ slots are already full, we insert a new triple as $(i, 1, 0)$. 
\item If $i$ is already recorded, we increase the corresponding counter by 1 and update the order in the summary.
\item Otherwise, we replace the last triple $(\text{item}_\ell, \text{count}_\ell, \text{overestimation}_\ell)$ with a new triple $(i, \text{count}_\ell + 1, \text{count}_\ell)$. 
\end{itemize}
After processing the stream we return as an estimate for an item weight either its counter in the summary or, if not recorded, the overestimation in the last counter in the list. For a stream of length $m$, the overestimation in a given counter is bounded by $m/\ell$ and it is guaranteed that an item occurring more than $m/\ell$ times will be in the summary. This comes from the fact that each increase in the overestimation of an item not recorded in the summary is witnessed by $\ell$ different items in the summary and this cannot happen more than $m/\ell$ times. However, the algorithm is known to perform extremely well in practice and report almost the exact weights for the heaviest entries, see e.g.~\cite{CormodeH09}.

A natural generalization of both algorithms works with weighted positive updates. This is straightforward for {\sc Count-Sketch}, and for {\sc Space-Saving} the only issue to consider is how to efficiently update the summary.  A solution  achieving amortized constant time per update is presented in~\cite{bisam}.

%%%%%%%%%%%%%%%%%%%%%%%%%%%%%%%%%%%%%%%%%%%%%%%%%%%%%%%%%%%%%%%%%%%%%%%%%%%%%%%%%%%%%%%%
%
%
%
% ADVERTISING THE ALGORITHM
%
%
%
%%%%%%%%%%%%%%%%%%%%%%%%%%%%%%%%%%%%%%%%%%%%%%%%%%%%%%%%%%%%%%%%%%%%%%%%%%%%%%%%%%%%%%%%

\section{Our approach}

\paragraph{The algorithm.}
Assume we are given matrices $A, B \in \mathbb{R}^{n \times n}$ such that $A$ is stored as column-major ordered triples and $B$ as row-major ordered triples. (For an overview of efficient implementations of the model we refer to Chapter 2 of~\cite{thesis}.)  The skeleton of our algorithm is the following:

\begin{itemize}
\item Define $s$ {\sc Count-Sketch}  estimators (or alternatively {\sc Space-Saving} summaries with $\ell$ triples, for a small constant $\ell$). 
\item For the $k$th column and row vectors $a_k \in A, b_k \in B$, $k \in [n]$, hash each entry in $ab$ to one of the estimators.
\item After all outer products have been processed return an estimate for all entries.
\end{itemize}

In the parallel version, each processor keeps track of a subset of the estimators, and the total space remains fixed.
Thus, we are in a {\em shared nothing\/} model with no need for a shared memory -- the only requirement is that each processor sees the column and row vectors for each outer product.
Pseudocode of the algorithm is given in Figure~\ref{main_alg}. We refer to it as the {\sc CRoP} algorithm, which refers to the fact that each processor crops the output matrix to produce a small fraction of it, and is also an acronym for ``Column-Row Parallel''.

A crucial property in our analysis and experimental evaluations is that the heaviest entries do not often collide, and thus we obtain high quality estimates on their weight.  
We combine two different ways for estimating the weight of the heaviest entries based on the {\sc Count-Sketch} and {\sc Space-Saving} algorithms. 
In particular, we use a distribution hash function $h: [n]\times [n] \rightarrow [\kappa]$ to split the
set of entries into $\kappa$ parts, and use a {\sc Space-Saving} sketch and a {\sc Count-Sketch} estimator on each part. 
The size $\kappa$ of the hash table and the size of the {\sc Space-Saving} sketch determines the accuracy of the estimates.

\paragraph{Parallel processing of entries.}
Na\"ively we could just iterate through all entries of each outer product, but we would like an algorithm that runs in time linear in the number of nonzero elements in the input vectors and the entries hashing to a given interval of the hash table. In other words, given sufficient parallelism we can handle a given data rate even if there is a huge number of entries in a given outer product.

%%%%%%%%%%%%%%%%%%%%%%%%%%%%%%%%%%%%%%%%%%%%%%%%%%%%%%%%%%%%%%%%%%%%%%%%%%%%%%%%%%%%%%%%
%
%
%
% KEY LEMMA 
%
%
%
%
%%%%%%%%%%%%%%%%%%%%%%%%%%%%%%%%%%%%%%%%%%%%%%%%%%%%%%%%%%%%%%%%%%%%%%%%%%%%%%%%%%%%%%%%
\begin{lemmx} \label{rtime_lemma}
Let $h_a, h_b: \mathcal{E} \rightarrow [\kappa]$ be pairwise independent hash functions. Given input vectors $a$ and $b$ with $O(t)$ nonzero entries each and an interval $\mathcal{L} = [q, q+1, \ldots, r)$ for $0 \le q \le r \le \kappa$ we can construct $\mathcal{E}_\mathcal{L}^{ab}$, the set of entries occurring in the outer product ${ab}$ hashing to a value in $\mathcal{L}$, in expected time $O(|\mathcal{E}_\mathcal{L}^{ab}| + t)$. 
\end{lemmx}
\begin{proof} 
%For a given pair $(u(i), v(j))$ in $O_{uv}$ we will abuse notation in the following and write it as $(i, j)$.
We exploit the special structure of our hash functions: $h(i,j)=({h_a}(i)+{h_b}(j)) \text{ mod } \kappa$ for a 2-wise independent hash function $h_a, h_b:  \mathbb{N} \rightarrow [\kappa]$ and $i, j \in [n]$. It is easy to show that this construction implies that $h: \mathcal{E} \rightarrow [\kappa]$ is also pairwise independent.  
To find the entries with a nonzero value in  ${ab}$ that hash to $[q,r)$ we first sort the indices with a nonzero value in each $a$ and $b$ according to the hash value in arrays $H_a$ and $H_b$ for $a$ and $b$, respectively.
Entries with a hash value in the right range correspond to elements in $H_a$ and $H_b$ with sum in $[q,r) \cup [\kappa+q,\kappa+r)$. Since the hash values are pairwise independent we can sort by bucket sort in expected linear time.
%We will denote $[k+i,k+j)$ by $[i,j)+k$.
One way to find the entries in the right range would be to iterate over elements of $H_a$, and for each do two binary searches in $H_b$ to find the values in the right ranges.
However, this can be further improved by processing $H_a$ in sorted order, and exploiting that the entries with hash value in the right range correspond to intervals of $H_b$ that will be moving monotonically left. Thus, for a fixed index $i$ we find the entries $(i, j)$, $j \in [n]$ with a hash value in $[q,r) \cup [\kappa+q,\kappa+r)$ in time linear in the number of such entries.
This brings down the time to $O(|\mathcal{E}_\mathcal{L}^{{ab}}| + t)$.
\end{proof}

\paragraph{Parameters.}
We will assume that data is not lightly skewed and $z>1/2$. We will distinguish between the cases when $1/2 < z < 1$ and $z>1$. In order to keep the presentation concise the particular case $z=1$ will not be analyzed. In our analysis we will also use as a parameter the total number of nonzero entries $d$ occurring in the matrix product.
The value of $d$ is not known in advance. One can either be conservative and assume $d = O(n^2)$, or in the case of positive input matrices use efficient methods for estimating the number of nonzero entries~\cite{ACP} if two passes over the input are allowed.

%%%%%%%%%%%%%%%%%%%%%%%%%%%%%%%%%%%%%%%%%%%%%%%%%%%%%%%%%%%%%%%%%%%%%%%%%%%%%%%%%%%%%%%%
%
%
%
%PSEUDOCODE
%
%
%
%%%%%%%%%%%%%%%%%%%%%%%%%%%%%%%%%%%%%%%%%%%%%%%%%%%%%%%%%%%%%%%%%%%%%%%%%%%%%%%%%%%%%%%%

\begin{figure}[t]
{ \footnotesize

\Kw{\sc CRoP}
\algsetup{indent=3em}
\begin{algorithmic}[1]
\REQUIRE matrices $A, B \in \mathbb{R^+}^{n\times n}$, Interval $[q, r)$,  a list $\kappa$ Space-Saving summaries $\mathcal{SS}$ each for $\ell$ entries, array for $\kappa$ real numbers $\mathcal{CS}$, pairwise independent hash functions $h_a, h_b: \mathbb{N} \rightarrow [\kappa]$, $s: \mathcal{E} \rightarrow \{-1,1\}$
\medskip
\FOR{$i \in [n]$}
\STATE $a:= i$th column of $A$, $b:= i$th row of $B$
\STATE $E := \mbox{\sc Hash}(h_a, h_b, a, b, [q, r))$ 
\FOR{$e \in E$}
	\STATE Update Space-Saving summary $\mathcal{SS}[h(e)]$ with the entry $e$. 
    \STATE $\mathcal{CS}[h(e)]=\mathcal{CS}[h(e)] + e\cdot s(e)$.
\ENDFOR
\ENDFOR

\end{algorithmic}
\bigskip
\Kw{\sc Hash}
\begin{algorithmic}[1]
\REQUIRE pairwise independent hash functions $h_a, h_b: \mathbb{N} \rightarrow [\kappa]$, vectors $a \in \mathbb{R}^n$, $b \in \mathbb{R}^n$, Interval $[q, r)$
\medskip
\STATE sort the nonzero entries in $a$ and $b$ according to $h_a$ and $h_b$ in arrays $H_a$ and $H_b$, respectively
\STATE   $k = H_b.length$, $E = \emptyset$
\FOR{$i=1$ to $H_a.length -1$}
\WHILE{$H_a[i] + H_b[k] > q$}
		\STATE if $k>0$: $k= k-1$
	\ENDWHILE
\STATE $j = k$
	\WHILE{$H_a[i] + H_b[j] < q$}
		\STATE $j = j+1$
		\STATE if $j >= H_b.length$: break for-loop
	\ENDWHILE 
	\WHILE{$H_a[i] + H_b[j] <r$}
	\STATE  $E = E \cup (i, j)$
	\STATE if $j < H_b.length$: $j = j + 1$
	\STATE else: break for-loop
	\ENDWHILE
\ENDFOR
\RETURN the set $E$
\end{algorithmic}
}
%\end{procedure}
%
%
%
%
%
%
\caption{ \label{main_alg} \footnotesize [CRoP algorithm] A high-level description for a single run of our algorithm. 
The description of the hash functions $h: \mathcal{E} \rightarrow [\kappa]$ and $s: \mathcal{E} \rightarrow \{-1,1\}$ is sent to all processors before the computation starts.
For a detailed pseudocode description of {\sc Space-Saving} and {\sc Count-Sketch} the reader is referred to the original works \cite{cs,spacesaving}. After processing the input matrices one can obtain an estimate for individual entries from the corresponding {\sc Space-Saving} summary or {\sc Count-Sketch} estimator.%Pseudocode for computing estimates is identical to that of {\sc Space-Saving} and {\sc Count-Sketch}, applied to each such structure. 
}  
\end{figure}

%%%%%%%%%%%%%%%%%%%%%%%%%%%%%%%%%%%%%%%%%%%%%%%%%%%%%%%%%%%%%%%%%%%%%%%%%%%%%%%%%%%%%%%%
%
%
%
% THEORETICAL ANALYSIS
%
%
%
%%%%%%%%%%%%%%%%%%%%%%%%%%%%%%%%%%%%%%%%%%%%%%%%%%%%%%%%%%%%%%%%%%%%%%%%%%%%%%%%%%%%%%%%

\section{Analysis}

%The proofs of the following theorems can be found in the appendix. 

\paragraph{Load balancing.}

First we show that pairwise independent hashing of entries guarantees good load balance among processors. We show that for outer products for which the number of non-zero entries is considerably bigger than the number of processors, {\sc CRoP} achieves good scalability with high probability.  %Each processor has to read all vectors in order to determine which entries are hashed to its corresponding buckets, thus we show that the probability for a deviation by more than the number of nonzero elements in the batch decreases with the number of processors. This guarantees that given sufficient parallelism with high probability the running time will not be dominated by a bad distribution among processors.  

\begin{thmx} \label{load_bal_thm}
Suppose we run $K$ instances of {\sc CRoP} with disjoint intervals of the same size, on an input of column and row vectors $a, b \in \mathbb{R}^n$. Let $W = \frac{|a||b|}{K}$. Then for $W \ge \frac{|a| + |b|}{\lambda^2}$ the probability that the number of entries processed by a given instance will deviate by more than $\lambda W$ from its expected value is bounded by $\frac{1}{(|a| + |b|)}$ for any $\lambda > 0$. %and $w_e$ the number of occurrences of an entry $e \in \mathcal{E}_m$ in the batch. Set $W_2 := (\sum_{e \in \mathcal{E}_m} w_e^2)^{\frac{1}{2}}$. Then the probability that the number of entries processed by a given instance will deviate by more than $W_2$ from its expected value is bounded by $1/K$.
\end{thmx}
\begin{proof}
We distribute the nonzero entries in $ab$ to $K$ processors. 
We consider one of the $K$ processors, say $P_1$.  Let $Y_e$ be an indicator random variable denoting whether the entry $e$ is hashed to the range of $P_1$. Clearly, $E[Y_e] = 1/K$. Let $X = \sum_{e \in ab} Y_e$. Thus, we have $E[X] = W$. Since the hash function is pairwise independent for the variance of $X$ we have $V[X] = \sum_{e \in ab} V[Y_e] = \sum_{e \in ab} (E[Y_e^2] - E[Y_e]^2) \leq \sum_{e \in ab} \frac{1}{K} = \frac{|a||b|}{K} = E[X] = W$. 
Using Chebyshev's inequality and $W \ge \frac{|a| + |b|}{\lambda^2}$ we estimate the probability that this number deviates by more than $\lambda$ to:
$$
  \Pr[|X-E[X]| \geq \lambda W]\leq \frac{V[X]}{\lambda^2 W^2} =$$ $$\frac{V[X]K^2}{(\lambda|a||b|)^2} = \frac{K}{\lambda^2|a||b|} \le \frac{1}{(|a| + |b|)} 
$$
\end{proof}

The above theorem essentially says that when we expect a sufficient number of entries per core in a given outer product, the probability that the work load will deviate by more than a small constant factor from its expectation, namely $W/K$, is very small. Thus, we expect that the computation of only a small fraction of the outer products will not be well balanced among processors and this will not considerably affect the total performance.

In the following we present results for the estimates based on \textsc{Space-Saving}
(giving guarantees on the upper/lower bounds), as well as the unbiased \textsc{Count-Sketch}
estimator returned by our algorithm. 
%Note that while {\sc Count-Sketch} is theoretically superior as it always returns an unbiased estimator and applies to arbitrary matrix products, it requires more space for high quality estimates and our evaluations show that it performs rather poorly on real datasets compared to {\sc Space-Saving}. 
%%%%%%%%%%%%%%%%%%%%%%%%%%%%%%%%%%%%%%%%%%%%%%%%%%%%%%%%%%%%%%%%%%%%%%%%%%%%%%%%%%%%%%%%
%
%
%
% SPACE SAVING
%
%
%
%%%%%%%%%%%%%%%%%%%%%%%%%%%%%%%%%%%%%%%%%%%%%%%%%%%%%%%%%%%%%%%%%%%%%%%%%%%%%%%%%%%%%%%%
\medskip
\paragraph{Quality of estimates.}

\begin{thmx}\label{cms_thm}
Let $A, B$ be $n \times n$ nonnegative matrices such that there are $d$ nonzero entries in $AB$ adhering to a Zipfian distribution with parameters $C$ and $z$. After processing all outer products $a_i b_i$, $1 \le i \le n$, by {\sc CRoP} with $k$ buckets and a Space-Saving data structure with $\ell$ entries
\begin{itemize}
\item if $z>1$, an entry with weight at least $\Omega(\frac{C}{(\ell k)^z})$ is recorded in one of the Space-Saving summaries with probability more than 1/2. By running $O(\log \frac{\ell k}{\delta})$ instances of {\sc CRoP} in parallel, we report all entries with weight $\Omega(\frac{C}{(\ell k)^z})$ or more with probability at least $1-\delta$.

\item if $z<1$, an entry with weight at least $\Omega(\max(\frac{C}{(\ell k)^z}, \frac{Cd^{1-z}}{\ell k}))$ is recorded in one of the Space-Saving summaries with probability more than 1/2. By running $O(\log \frac{K}{\delta})$ instances of {\sc CRoP} in parallel, for $K = \max(\ell k, \frac{d}{\ell k})$, we report all entries with weight  $\Omega(\max(\frac{C}{(\ell k)^z}, \frac{Cd^{1-z}}{\ell k}))$ or more with probability at least $1-\delta$.
\end{itemize}
\end{thmx}
\begin{proof}
Let the minimum absolute weight for heavy entries be $\alpha m$, $\alpha > 0$. (At the end of the proof we will obtain bounds on $\alpha m$ depending on $k$.)
We will estimate the probability that a heavy entry $e$ is not reported.
From the Zipfian distribution we obtain that $x := (\frac{C}{\alpha m})^{\frac{1}{z}}$ entries
will have weight above $\alpha m$. Let $\B$ be the bucket $e$ hashes to.
We first consider the case that at most $\ell/2$ entries with weight above $\alpha m$ are hashed to $\B$. The expected number of heavy entries hashed to $\B$ is $x/k$. For $k \geq \frac{12x}{\ell}$ we get by Markov's inequality that the probability $p_1$ that more than $\ell/2$ heavy entries will land in $\B$ is at most $1/6$. 
As already discussed in Section~\ref{sec:prel}, if the total weight of non-heavy entries hashed to $\B$ is
less than $\ell\alpha m/2$ and no other heavy entries hash to $\B$, then all heavy entries will be reported. Let $w := \sum_{i=x+1}^d {C/i^z}$ be the total weight of non-heavy entries. 
In the following we will use the fact that $w = O(C(d^{1-z}))$ for $z<1$ and $w = O(C(x^{1-z}))$ for $z> 1$, where $d$ is the number of distinct entries in the matrix product.
This follows by integration of the corresponding continuous function. %We also define $w_2 := \sum_{i=x+1}^d {(C/i^z)^2}$  
Then we expect non-heavy entries of total weight $w/k$ to land in $\B$.
%In order to show a a small deviation from the expected value we will adapt the analysis
%from~\cite{cms-zipf} to our problem. 

Let $Y_{j}$ be an indicator random variable denoting whether the $j$th non-heavy entry is
hashed to $\B$ and $X = \sum_{j=1}^{d-x}Y_j$.
Clearly, $E[X] = w/k$. 
Applying Markov's inequality we obtain $p_2 := \Pr[X \geq 3 {w}/{k}] \leq 1/3$. %We will later enforce $p_2 \leq 1/3$.
We want $3{w}/{k} \leq \ell\alpha m/2$.
For $z<1$ we have $w = O(C(d^{1-z})$ thus together with the bound on the number of heavy entries in the bucket we set $k = \max(O(
\frac{1}{\ell}(\frac{C}{\alpha m})^{\frac{1}{z}}), O(\frac{C(d^{1-z})}{\ell\alpha m)})$.
Similarly, for $z>1$ we have $w = O(C(x^{1-z})$, hence the bound
$k = O(\frac{1}{\ell}(\frac{C}{\alpha m})^{\frac{1}{z}})$. Thus, for $z<1$, we will consider entries {\em heavy} if their weight is $\alpha m = \Omega(\max(\frac{C}{(\ell k)^z}, \frac{Cd^{1-z}}{\ell k}))$ and for $z>1$ if $\alpha m = \Omega(\frac{C}{(\ell k)^z})$. By the union bound the probability that an entry with weight at least
$\alpha m$ is not reported is at most $p_1+p_2 < 1/2$.

By running $t$ copies of the algorithm in parallel and reporting only entries that are reported in some summary in at least $t/2$ cases, we can amplify the probability for a correct estimate to $\exp(O(t))$. The analysis follows a standard application of Chernoff's inequality. Thus, for the given values of $t$ and $K$ for each heavy entry we bound the probability of not being reported to $\frac{\delta}{\ell k}$ and $\frac{\delta}{K}$ for $z>1$ and $z<1$, respectively. Since the number of heavy entries in each case is bounded by $\ell k$ or $K$, by the union bound we bound the error to $\delta$ that any heavy entry will not be reported.
\end{proof}
%
%
%
%
%
%
%
%
%
%%%%%%%%%%%%%%%%%%%%%%%%%%%%%%%%%%%%%%%%%%%%%%%%%%%%%%%%%%%%%%%%%%%%%%%%%%%%%%%%%%%%%%%%
%
%
% Count Sketch
%
%
%
%%%%%%%%%%%%%%%%%%%%%%%%%%%%%%%%%%%%%%%%%%%%%%%%%%%%%%%%%%%%%%%%%%%%%%%%%%%%%%%%%%%%%%%%
\medskip
{\sc Count-Sketch} was analyzed for Zipfian data in the original publication~\cite{cs}. 

\begin{thmx}~\cite{cs}\label{cs_thm}
Let $A, B \in \mathbb{R}^{n \times n}$ such that the entries in $AB$ adhere to a Zipfian distribution with parameters $C$ and $z>1/2$. After processing all outer products $a_k b_k$ by {\sc CRoP} with $k$ {\sc Count-Sketch} estimators, each entry is approximated with additive error of at most $\Omega(\frac{C}{k^z})$ with probability more than $1/2$. By running of $O(\frac{n}{\delta})$ instances in parallel the additive error of at most $\Omega(\frac{C}{k^z})$ for any entry  holds with probability at least $1-\delta$.  
\end{thmx}

%%%%%%%%%%%%%%%%%%%%%%%%%%%%%%%%%%%%%%%%%%%%%%%%%%%%%%%%%%%%%%%%%%%%%%%%%%%%%%%%%%%%%%%%
%
%
% LOWER BOUND
%
%
%
%%%%%%%%%%%%%%%%%%%%%%%%%%%%%%%%%%%%%%%%%%%%%%%%%%%%%%%%%%%%%%%%%%%%%%%%%%%%%%%%%%%%%%%%

\paragraph{A lower bound.}%\label{sec:lower}

We present a lower bound for the space needed to report the heaviest entry in a
matrix product by the column-row method. 
Of course, this applies also to any harder problem, such as reporting a larger set of heavy entries, with estimates.
%This lower bound complements a strong worst-case space lower bound presented in~\cite{mining_stream}, by arguing that {\em any\/} data streaming algorithm for frequent pairs that does not make errors must store all information about itempair counts.

\begin{thmx} \label{lb_thm}
Any algorithm that always outputs the heaviest entry in a matrix product $AB$, or even just the weight of the heaviest entry, in only one pass over the outer products $a_i b_i$, $i \in [n]$, must encode in its state all entry weights, in the sense that if two prefixes of the outer products differ in the weight of some entry, the algorithm must be in different states after seeing the prefixes.
\end{thmx}

\begin{proof} 
Since we are allowed only one pass over the column and row vectors of the input matrices, we will consider the outer products as a stream of updates for the output entries.
For a prefix of the stream consider the count vector that for each
entry records its weight in the prefix.
Let $\mathcal{A}$ be any algorithm that computes the heaviest entry in a data stream.
Consider two distinct weight vectors $x$ and $y$, corresponding to different stream prefixes.
We argue that $\mathcal{A}$ must
be in two different states after seeing these prefixes.
Suppose that the latter claim is not true, so that for $x\neq y$ the algorithm is in the same
state.
Since $x$ and $y$ differ, there must be at least one entry $(i,j)$ having distinct weights in the
two weight vectors.
But this implies that we can extend the streams with a sequence that makes the entry $(i,j)$ the heaviest
one in one of the weight vectors, say w.l.o.g. $x$, $(i,j)$ becomes the heaviest entry, while this
does not happen in $y$.
Still, the algorithm would be in the same state in both cases, returning the same result and weight.
This contradicts the assumption that $\mathcal{A}$ always returns the correct answer, so the assumption that $x$ and $y$ resulted in the same state must be false.
\end{proof}

Intuitively, the only generally applicable ways of storing the weights of all entries is to either store each weight explicitly, or store all column and row vectors seen so far.
%Indeed, in the full version of this paper we show that this means that for a wide range of input distributions, after seeing a sufficiently long prefix of the outer products any algorithm based on the column row method needs space $\Omega(t^2)$ (with high probability), where $t$ is the number of nonzero entries in the column and row vectors.

%\begin{proof}
%We will consider the outer products as a stream of updates for the output entries.
%For a prefix of the stream consider the count vector that for each
%entry records its weight in the prefix.
%Let $\mathcal{A}$ be any algorithm that computes the heaviest entry in a data stream.
%Consider two distinct weight vectors $x$ and $y$, corresponding to different stream prefixes.
%We argue that $\mathcal{A}$ must
%be in two different states after seeing these prefixes.
%Suppose that the latter claim is not true, so that for $x\neq y$ the algorithm is in the same
%state.
%Since $x$ and $y$ differ, there must be at least one entry $(i,j)$ having distinct weights in the
%two weight vectors.
%But this implies that we can extend the streams with a sequence that makes the entry $(i,j)$ the heaviest
%one in one of the weight vectors, say w.l.o.g. $x$, $(i,j)$ becomes the heaviest entry, while this
%does not happen in $y$.
%Still, the algorithm would be in the same state in both cases, returning the same result and weight.
%This contradicts the assumption that $\mathcal{A}$ always returns the correct answer, so the assumption that $x$ and $y$ resulted in the same state must be false.
%\end{proof}

\section{Frequent pair mining in transactional data streams: a case study}

Our original motivation was to parallelize the mining of frequent pairs from a high speed transaction stream. As already explained, the problem can be seen as an instance of sparse matrix multiplication $AA^T$ by the column row method. 

We use the following notation for transactional data streams. Let $\mathcal{I}$ be a set of items. A {\em transaction} $T \subset \mathcal{I}$ is a set of items.  We call a subset $p=\{i,j\} \subset \mathcal{I}$ a {\em pair}. For a stream of transactions $\mathcal{S}$, the {\em support} of a pair $p$ is the number of transactions containing $p$: $\sup(p) = |\{T: p \subseteq T\}|, T \in \mathcal{S}$. We can consider each transaction $T$ as a $\{0,1\}$-valued sparse $n$-dimensional vector $v_T$ and the set of pairs occurring in the given transaction as the nonzero entries above the diagonal in the outer product of $v_T$ and its transpose.

\subsection{Frequent Pattern Mining.}

%Our original aim was to design a parallel algorithm for the  fundamental task in knowledge discovery of mining association rules from transactional data sets. A preliminary report of our work presenting some experimental results but no theoretical guarantees can be found in~\cite{KDCloud}.
%The computationally expensive step is discovering the most frequent pairs~\cite{parketal}, i.e. pairs occurring in a certain number of transactions.
%This problem, in turn, is equivalent to finding the largest entries in a product $A^{T}A$ where $A$ is a 0-1 matrix, typically sparse.

Considerable work has been done on parallel and distributed implementations of frequent itemset methods. 
We refer to Zaki's 1999 survey~\cite{Zaki99} for an overview of some main techniques that have been investigated.
Subsequent work has focused on
%\begin{enumerate}
either multi-core computing, aiming to minimize the overhead of access to shared data structures~\cite{conf/vldb/GhotingBPKNCD05,conf/vldb/LiL07}, or more recently GPU computing, focusing on representations of data that allow efficient GPU implementations (see e.g.~\cite{conf/fskd/LiZ10,GPU-FPM} for recent results).
%\end{enumerate}

To our best knowledge, all known methods for parallelising frequent itemset mining use either a shared data structure, or a ``vertical'', column-by-column layout of the data (for each item, a sorted list of the transactions where it occurs).
In the former case the shared data structure will become a bottleneck when scaling to many cores.
In the latter case one can no longer have a space-efficient streaming algorithm, because computing the vertical representation requires storing all transactions.
It is also well-known that for sparse matrices this representation, which supports fast computation of inner products, leads to higher time complexity than methods based on outer products. 
See~\cite{AmossenPaghICDT} for an overview of theoretical results on (serial) sparse matrix products.

%%%%%%%%%%%%%%%%%%%%%%%%%%%%%%%%%%%%%%%%%%%%%%%%%%%%%%%%%%%%%%%%%%%%%%%%%%%%%%%%%%%%%%%%%%%
%%%%%%%%%%%%%%%%%%%%%%%%%%%%%%%%%%%%%%%%%%%%%%%%%%%%%%%%%%%%%%%%%%%%%%%%%%%%%%%%%%%%%%%%%%%

\subsection{Related work on stream mining.}%\label{sec:related}

Though we view the new load balancing technique as the main contribution of our work, its application to stream mining is interesting in its own right. 
In the following we briefly review related work in this area.

Manku and Motwani~\cite{manmot} first recognized the necessity for efficient
algorithms targeted at frequent itemsets in transaction streams and presented a heuristic approach generalizing  
their {\sc StickySampling} algorithm. A straightforward approach to mining of frequent pairs %(mentioned, but dismissed, in~\cite{manmot,mining_stream}) 
is to reduce the problem to that of mining frequent items by generating all item pairs in a given transaction.
Yu et al.~\cite{yuetal} and Campagna and Pagh~\cite{mining-stream} present randomized algorithms for transaction stream mining. The theoretical bounds on the quality of their estimates however heavily depend on the assumption that transactions are either generated {\em independently at random\/} by some process or arrive in a random order. It is already clear from the experiments of~\cite{mining-stream}
that such optimistic assumptions do not hold for many data sets.
For both schemes~\cite{yuetal,mining-stream} it is easy to find an ordering of essentially any transaction stream that breaks the randomness assumption, and makes it perform much worse than the theoretical bounds.
%Figure~\ref{fig:accidents} gives further evidence that this assumption can be problematic even if the ordering is not adversarial.
We therefore believe that a more conservative model is needed to derive a rigorous theoretical analysis, while exploiting observed properties of real data sets.

\subsection{Experimental evaluation.}

We worked with a cache-optimized Java implementation working only with primitive data types and used the built-in random number generator to store hash values in a table. Unless otherwise reported we ran experiments on as Mac Pro desktop equipped with Quad-Core Intel 2.8GHz and 16 GB main memory. In this architecture there are 8 cores and a total of 24 MB cache available but 2 cores share 6 MB of cache.
\\
\paragraph{Datasets.}

We evaluated the performance of our algorithm on the following datasets: Mushroom, Pumsb, Pumsb\_star and Kosarak, taken from the Frequent Itemset Mining Implementations (FIMI) Repository, Wikipedia -- crafted according to what is described in~\cite[Page~14]{msc-students}, and Nytimes and Pubmed  taken from the UCI Machine Learning Repository.

Table~\ref{table_datasets} summarizes the data sets used in experiments. In all cases, we use the order in which the transactions are given as the stream order.  %See Appendix~\ref{exper} for estimates on {\sc Count-Sketch}. 

\begin{table}
{ \footnotesize
%\begin{center}
\centering
\begin{tabular}{c|c|c}%p{17mm}|p{18mm}|}
{\bf Dataset} & {\bf \# of pairs ($F_2$)} & {\bf \# of distinct pairs} \\ 
\hline\hline
%Chess & $22.5\cdot 10^5$ & $2.66\cdot10^3$\\
%Connect & $639\cdot 10^5$ & $6.96\cdot10^3$\\
Mushroom & $22.4\cdot 10^5$ & $3.65\cdot10^3$\\
\hline
Pumsb & $1360\cdot 10^5$ & $536\cdot10^3$\\
\hline
Pumsb\_star & $638\cdot 10^5$ & $485\cdot10^3$\\
\hline
Kosarak & $3130\cdot 10^5$ & $33100\cdot10^3$\\
\hline
%BMS-WebView-1 & $9.64\cdot 10^5$ & $64.5\cdot10^3$\\
%BMS-WebView-2 & $24.4\cdot 10^5$ & $725\cdot10^3$\\
%BMS-POS & $246\cdot 10^5$ & $381\cdot10^3$\\
Retail & $80.7\cdot 10^5$ & $3600\cdot10^3$\\
\hline
Accidents & $187\cdot 10^5$ & $47.3\cdot10^3$\\
%T10I4D100K & $62.8\cdot 10^5$ & $171\cdot10^3$\\
%T40I10D100K & $841\cdot 10^5$ & $433\cdot10^3$\\
\hline
Webdocs & $2.0\cdot 10^{11}$ & $> 7\cdot 10^{10}$ \\
\hline
Nytimes & $1.0\cdot 10^{10}$ & $> 5\cdot 10^{8}$ \\
\hline
Pubmed & $1.6\cdot 10^{10}$ & $> 6\cdot 10^{8}$\\
\hline
Wikipedia & $5.17\cdot 10^{11}$ & $> 5.8\cdot 10^{9}$ \\
\hline
\end{tabular}
%\end{center}
\caption{\footnotesize Information on data sets for our experiments.
Nytimes and Pubmed are taken from the UCI Machine Learning Repository (Bag of Words data set).
The wikipedia dataset has been crafted according to what is described in~\cite[Page~14]{msc-students}. The other data sets are from the Frequent Itemset Mining Implementations (FIMI) Repository.
For the last three data sets the number of distinct pairs was estimated using a hashing technique
from~\cite{ACP}. In the context of sparse matrix multiplication the number of pairs is the total number of nonzero entries in all outer products and the number of distinct pairs is the number nonzero entries in the output matrix.}\label{table_datasets}
}
\end{table}
\paragraph{Accuracy of results.}

Our first set of experiments shows results on the precision of the counts obtained by CRoP using a {\sc Space-Saving} data structure of size 2,i.e., we record only two pairs.
%, as well as the recall.
The accuracy depends on the amount of space used and the number of pairs we are interested in reporting. Assume the pairs are sorted in decreasing order according to their frequency, we say that $i$th pair has rank $i$.
We made one experiment fixing the space usage, and looking at results for pairs of decreasing rank (computed exactly), and one that varies the space usage and considers the top-100 pairs. In practice, it may be hard to foresee how much space will be needed for a particular stream, so probably one will tend to use as much space as feasible with respect to running time (ensure in-cache hash table), or what amount of memory can be made available on the system.
A consequence of this will be even more precise results.
The results of our experiments on the Nytimes data set can be seen in Figure~\ref{fig:ss_estimates}(a). 

\paragraph{Varying space usage} 
We now investigate what happens to the quality of results when the space usage of CRoP is pushed to, and beyond, its limits.
For this, we chose to work with 3 representative data sets, namely Mushroom, Retail and Accidents, for decreasing space usage,
plotting the ratio between the upper and lower bounds for the top-100 pairs returned by our algorithm.
This is shown in Figure~\ref{fig:ss_estimates}(b) and we can see how the transition between very poor and very good quality is fairly
fast.

\begin{figure}[ht]
\begin{center}
\subfigure{
\includegraphics[scale=0.3]{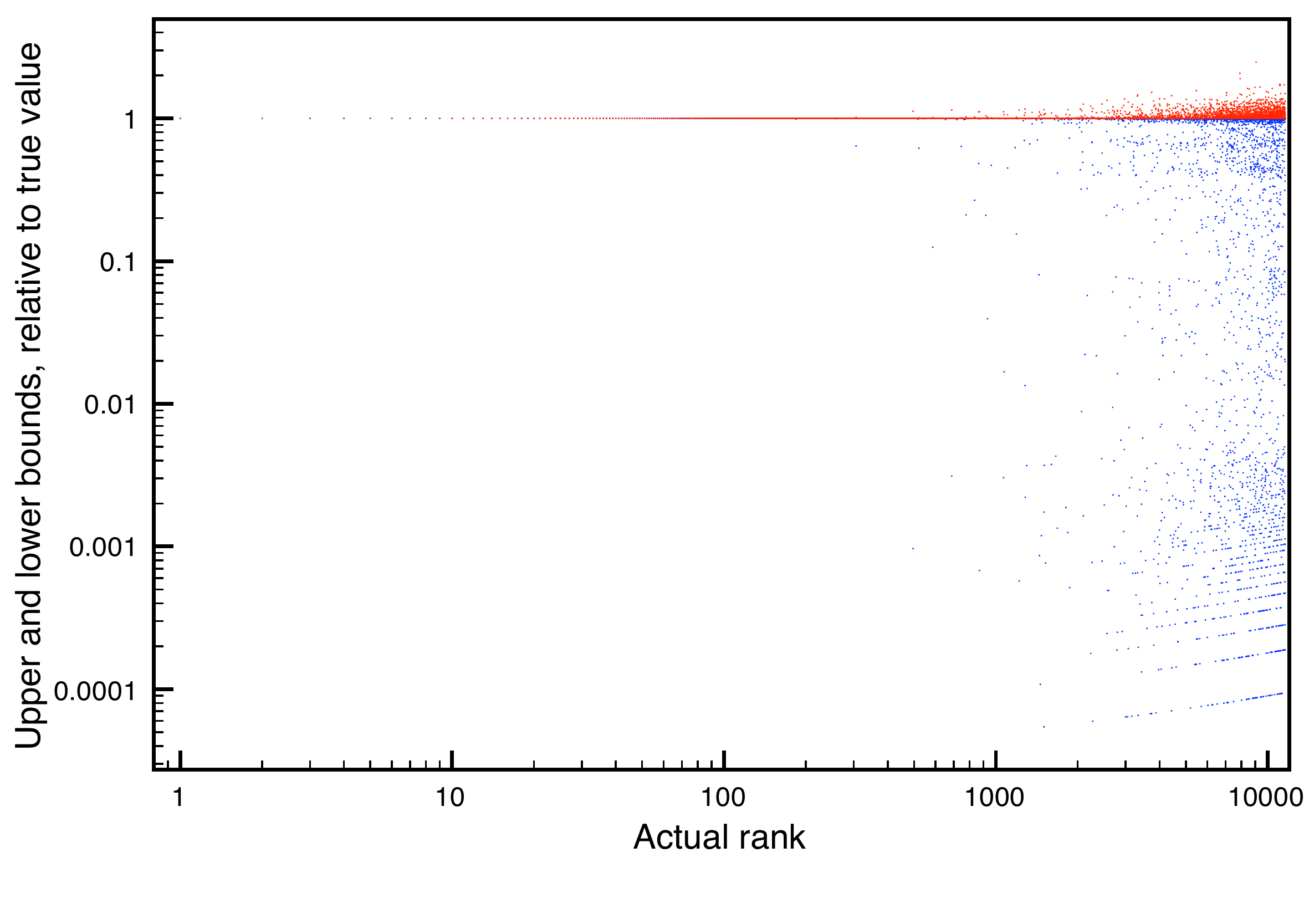}

}
\subfigure{
\includegraphics[scale=0.5]{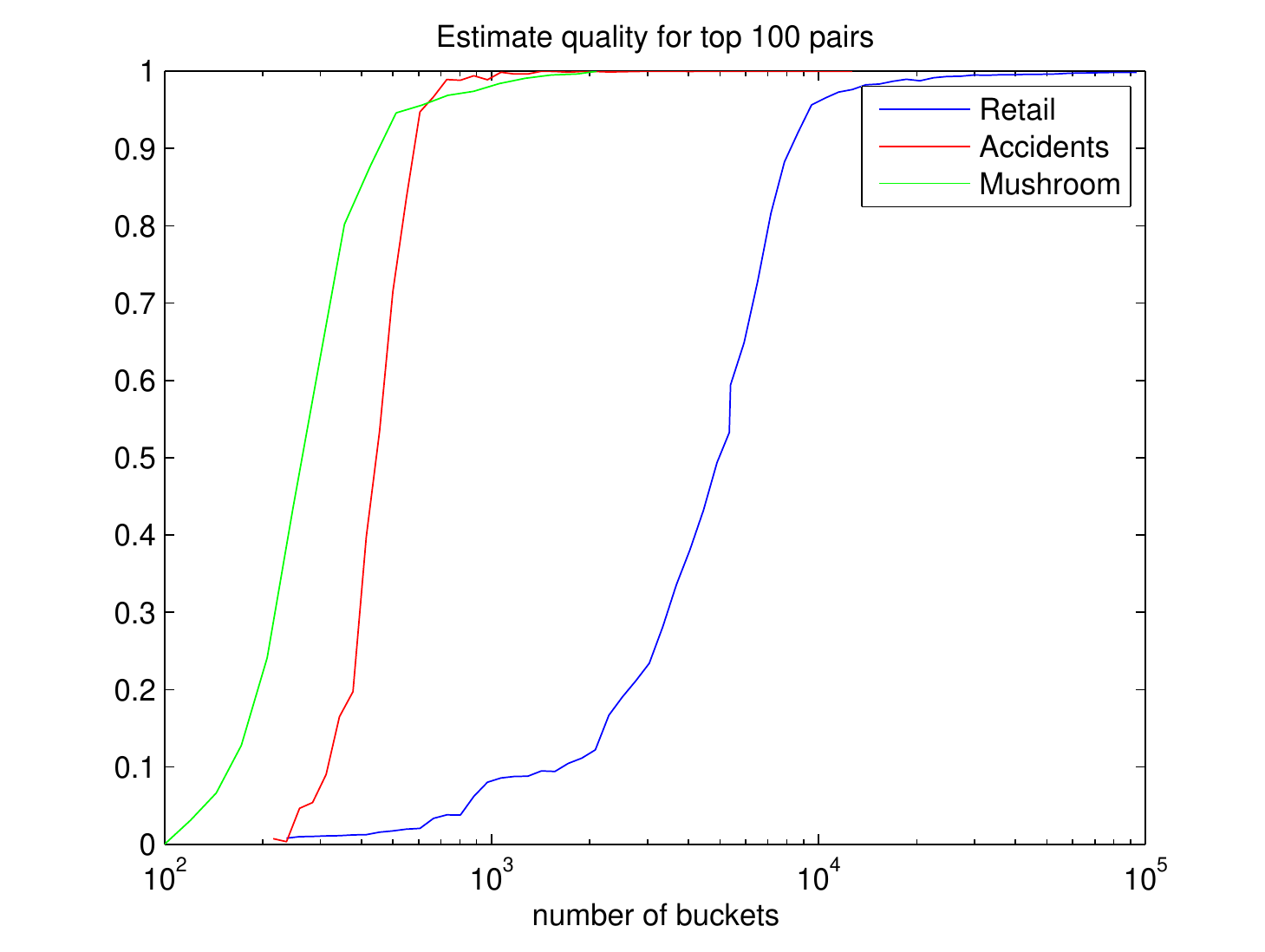}

}
\vspace{-3mm}
\end{center} 
\caption{\footnotesize (a) The plot on the left side shows upper and lower bounds for Nytimes computed by CRoP using $10^{6}$ buckets. Values are normalized by dividing by true support. Upper bounds shadow lower bounds, exact bounds are visible only as a red dot with no blue dot below. As can be seen, upper bounds are generally tighter than lower bounds. (b) On the right-hand side we plot the average ratio of lower and upper bound for top-100 pairs, for three representative data sets, as function of number of buckets. As can be seen, there is a quick transition from poor to excellent precision.}
\label{fig:ss_estimates}
\end{figure} 

\paragraph{Larger Space-Saving data structures.}
We estimated
the number of pairs with ratio at least 90\% between lower
and upper bound by varying the number of pairs 
in each bucket but keeping the total number of pairs of
our sketch constant. We counted the number of pairs until
the 100-th bad estimate, i.e. until 100 pairs have a lower
bound less than 90\% of its upper bound. Since our algorithm
is randomized we chose 100 as cut-off in order not to be
sensitive to outliers in the estimates. We fixed the total
number of pairs for the Kosarak dataset to 240000 and for
the Accidents dataset to 6000. We then varied the pairs per
bucket from 2 to 24. The effect was much better estimates
with the same space usage, see Figure~\ref{fig:ss}.

\paragraph{Count-Sketch Estimates} 
The result of the unbiased
COUNT-SKETCH estimator for the Kosarak dataset with
50000 buckets and 2 pairs per bucket is presented in Figure~\ref{fig:cs}.
We ran the algorithm 11 times and for each pair reported at
least 6 times we return the median of its estimates. The plot
shows the ratio of our estimates and the exact count of the
3000 pairs with highest support computed by Apriori. Not
reported pairs have ratio 0. In general we observe that the estimates given by {\sc Space-Saving} are tighter.

%\begin{figure}
%\begin{center}
%\includegraphics[width=90mm]{cs_kosarak.pdf}
%\end{center}
%\caption{Ratio of estimates and true count for the top 3000
%pairs of Kosarak. All top 1400 pairs are reported by our
%algorithm and for most of the pairs the estimates are within
%of factor 2. However, the estimates are worse than those given by Space-Saving with the same space requirements.}\label{fig_cs}
%\end{figure}

\begin{figure}[ht]
\begin{center}
\subfigure{
\includegraphics[scale=0.38]{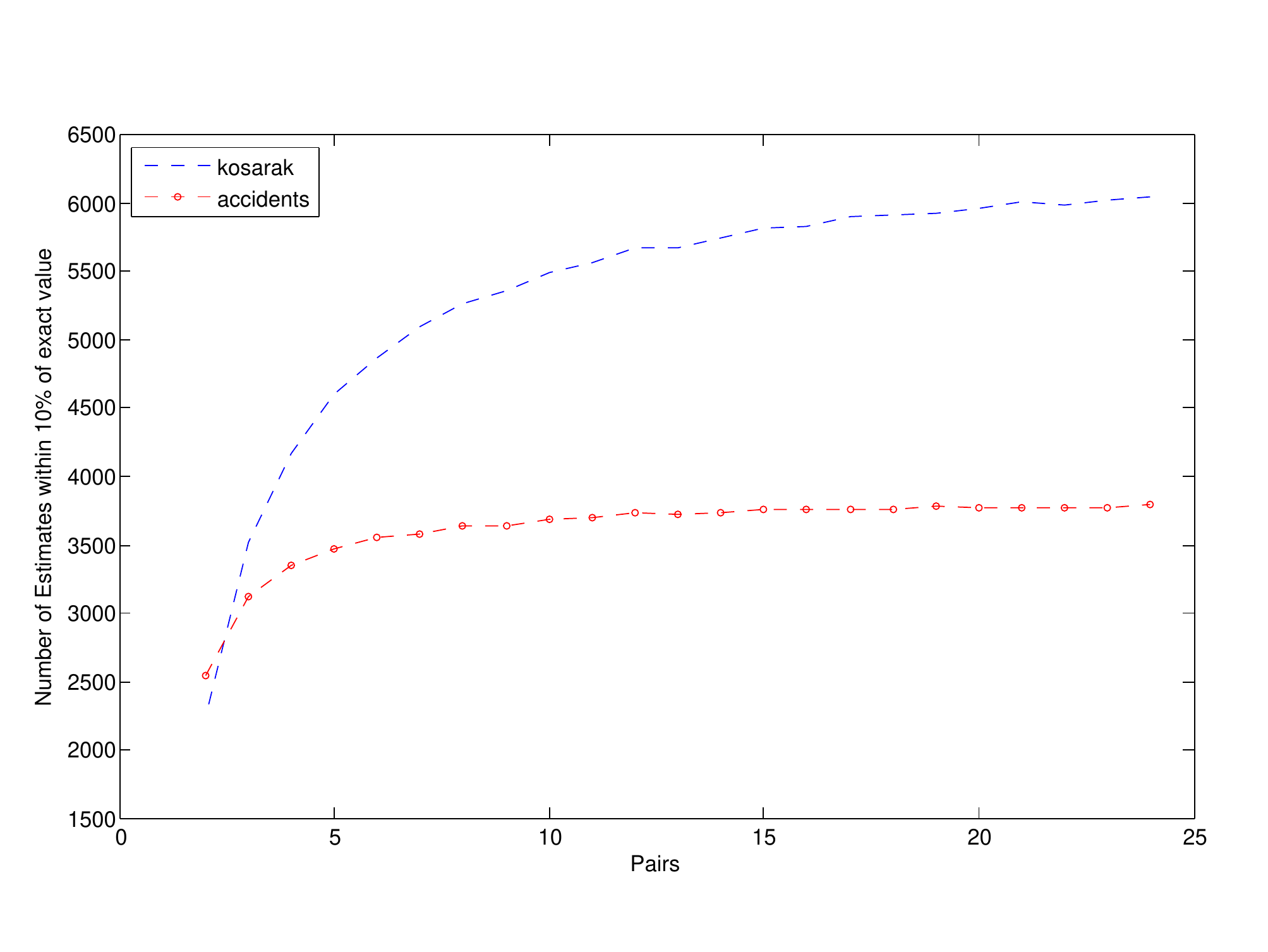}
\label{fig:ss}
}
\subfigure{
\includegraphics[scale=0.44]{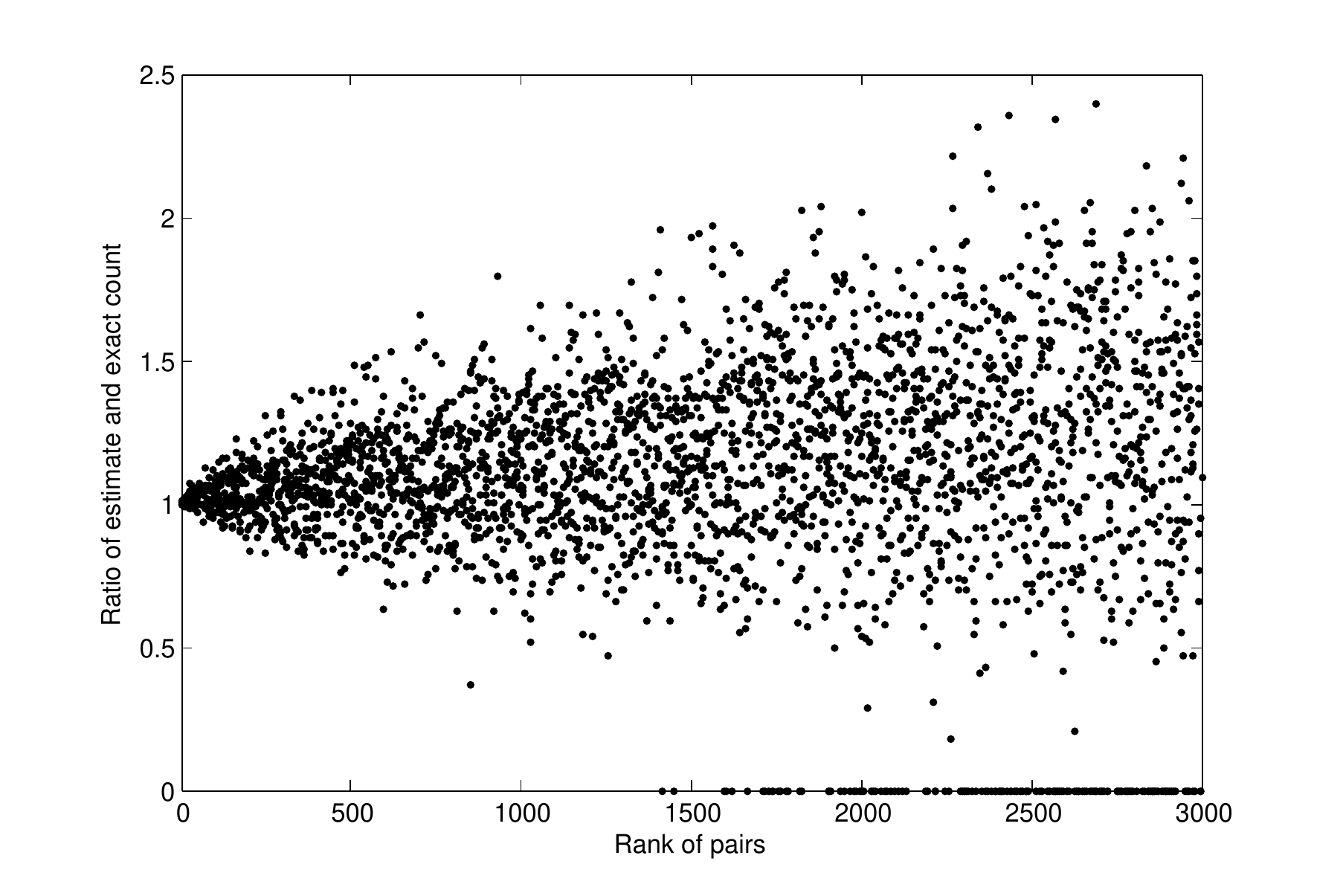}
\label{fig:cs}
}
\vspace{-3mm}
\end{center}
\caption{\footnotesize 
On the left side is the number of estimates with ratio of at least 90\% 
until 100 bad estimates have been
seen for a given number of pairs in each Space-Saving summary. As we see a sharp improvement in the quality of the
estimates can be observed by increasing the number of pairs
in each summary. 
On the right-hand side the plot represents the ratio of estimates and true count for the top 3000 pairs of Kosarak. All top 1400 pairs are reported
by our algorithm and for most of the pairs the estimates are within of factor 2. However, the estimates are
worse than those given by Space-Saving with the same total space usage.}
\end{figure}

\paragraph{Scalability.} In order to assess the scalability of our approach, we performed the following experiment: We ran the algorithm with a {\sc Space-Saving} data structure on the Kosarak dataset with a sketch of size $10000 i$ on $i$ cores. For each run we computed the recall for the top 10000 pairs in the dataset, see Table~\ref{table:kos_recall}. The idea is that each core can handle a sketch of size 10000 using only its local cache. The running time is the time needed for all processes to complete. Given that each core needs about 8 sec to read all input transactions, sort it according to the hash values of the items and then decide which pairs it is responsible for, a very good scalability is observed. The results indicate that using parallelism our approach is both more efficient and achieves more accurate results. 

\begin{table} 
{ \footnotesize
%\begin{center}
\centering
\begin{tabular}{c || c | c | c | c}
Cores & {\hspace*{8mm} 1} & {\hspace*{8mm} 2} & {\hspace*{8mm} 4} & {\hspace*{8mm}   8}\\ 
\hline 
Time(sec) & {\hspace*{6mm}22.13} & {\hspace*{6mm} 14.87} & {\hspace*{6mm} 11.73} & {\hspace*{6mm} 9.96}\\
\hline
Recall & {\hspace*{5mm}0.1119} & {\hspace*{5mm} 0.1988} & {\hspace*{5mm} 0.3156} & {\hspace*{5mm} 0.4812}\\
\hline
\end{tabular}
\caption{Recall for Kosarak.}  \label{table:kos_recall}
} 
\end{table}

\paragraph{Load balance.}

%\paragraph{Balancing of pairs}
We have run experiments evaluating the distribution of pairs among the buckets.
When running these experiments, we kept track of the number of pairs processed by each core.
The results of these experiments are reported in Table~\ref{fig:loadBal}.
The numbers in the table confirm that the pairs spread evenly amongst the buckets 
meaning that parallelism greatly improves the running time of the algorithm, since there
will be no core that has to sustain a much larger burden than the others; such a negative
situation would bring the performances of the algorithm close to a sequential one.
\begin{table}
\footnotesize{
 \begin{center}
  \begin{tabular}{c|c|c|c|}\\%c|}
   {\bf Dataset} & {\bf Cores} & {\bf Average} &  {\bf Maximum}\\% & {\bf Min.} \\
   \hline
   \multirow{3}{*}{Retail}  &       8     & $895542$ & $934040$\\%  & $874744$  \\
              {}    &       4        & $1791084$  & $1808784$\\%      & $1777085$   \\
              {}    &       2        & $3582168$ & $3585869$\\%     & $3578466$  \\
   \hline
   \multirow{3}{*}{Kosarak} &  8        & $3203546$ & $3219769$\\%  & $3191070$\\
%   \multirow{3}{*}{Docword.kos} &  8        & $3.2035\cdot10^{6}$ & $3.2198\cdot10^{6}$  & $3.1911\cdot10^{6}$\\
              {}    &       4        & $6407091$ & $6427398$\\%     & $6390764$\\
%              {}    &       4        & $6.407\cdot10^{6}$ & $6.4274\cdot10^{6}$  & $6.3908\cdot10^{6}$\\
%              {}    &       2        & $12814182$ & $12826728$ & $12801636$\\
              {}    &       2        & $1.2814\cdot10^{7}$ & $1.2827\cdot10^{7}$\\% & $1.2802\cdot10^{7}$\\
   \hline
%  
%   \multirow{3}{*}{Webdocs} & 8   & $892799048$ & $893941752$ & $891502398$\\
   \multirow{3}{*}{Webdocs} & 8   & $8.928\cdot10^{8}$ & $8.9394\cdot10^{8}$\\% & $8.915\cdot10^{8}$\\
%              {}    &       4        & $1785598095$ & $1787122525$  & $1783799210$\\
              {}    &       4        & $1.7856\cdot10^{9}$ & $1.7871\cdot10^{9}$\\%  & $1.7838\cdot10^{9}$\\
%              {}    &       2        & $3571196190$ & $3572984577$ &  $3569407803$\\
              {}    &       2        & $3.5712\cdot10^{9}$ & $3.573\cdot10^{9}$\\% &  $3.5694\cdot10^{9}$\\
   \hline
%   \multirow{3}{*}{Nytimes} & 8        & $1249737211$ & $1251028881$ & $1248636000$\\
   \multirow{3}{*}{Nytimes} & 8        & $1.2497\cdot10^{9}$ & $1.251\cdot10^{9}$\\% & $1.2486\cdot10^{9}$\\
%              {}    &       4        & $2499474421$ & $2500495223$  & $2499022815$ \\
              {}    &       4        & $2.4995\cdot10^{9}$ & $2.5005\cdot10^{9}$\\%  & $2.499\cdot10^{9}$ \\
%              {}    &       2        & $4998948843$ & $4999518038$ &  $4998379647$\\
              {}    &       2        & $4.9989\cdot10^{9}$ & $4.9995\cdot10^{9}$\\% &  $4.9983\cdot10^{9}$\\
   \hline
   Wikipedia & 8        & $6.4582\cdot10^{10}$ & $6.4651\cdot10^{10}$\\
% & $6.4496\cdot10^{10}$\\
%
%                        2        & $3.95\cdot10^{9}$ & $4.1\cdot 10^{5}$ &  $1.9\cdot10^{9}$ & $9.2\cdot10^{4}$\\
   \hline
   \end{tabular}\caption{\footnotesize The table shows the average and maximum number of pair occurrences handled by each core. As can be seen, the maximum is quite close to the average.}\label{fig:loadBal}
 \end{center}
 }
\end{table}

\paragraph{Progress of processing.}

%To illustrate the load balancing in terms of running time, we ran CRoP on 4 cores of an Intel Core(TM) i7 Q720 1.60GHz  Linux machine.
We ran four processes, and let the operating system allocate one to each core.
We tracked the progress of each process over time. 
The result for the Kosarak dataset are shown in Figure~\ref{fig_kosarak-progress} and for Nytimes dataset  in Figure~\ref{fig_nytimes-progress}.
As can be seen, the cores make almost identical progress when running at full speed.
In a data streaming setting this means that we can expect to manage streams that require all cores to run at close to maximum speed, while needing to cache only a small number of transactions.

\begin{figure}[ht]
\begin{center}
\subfigure{
\includegraphics[scale=0.32]{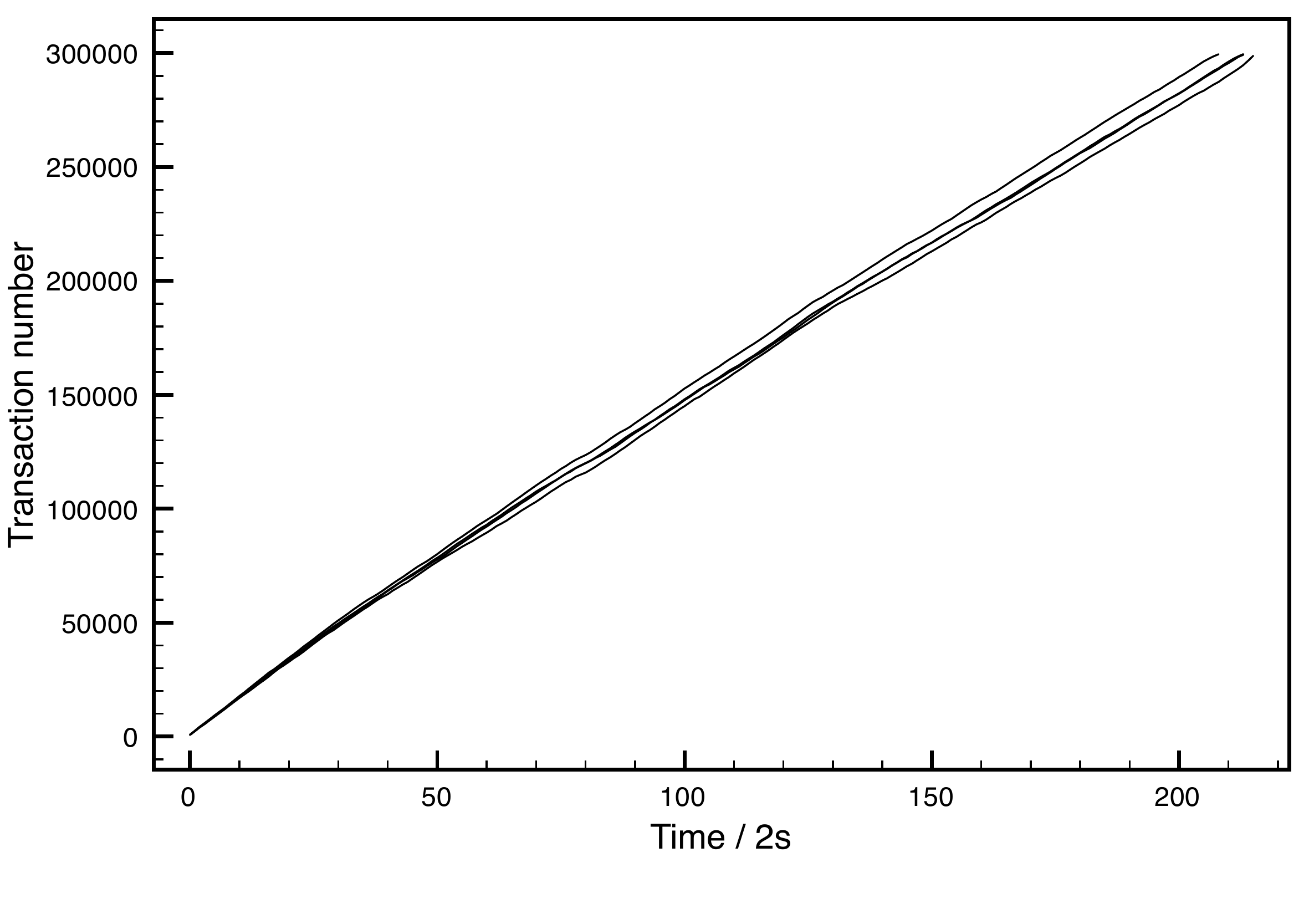}
\label{fig_nytimes-progress}
}
\subfigure{
\includegraphics[scale=0.32]{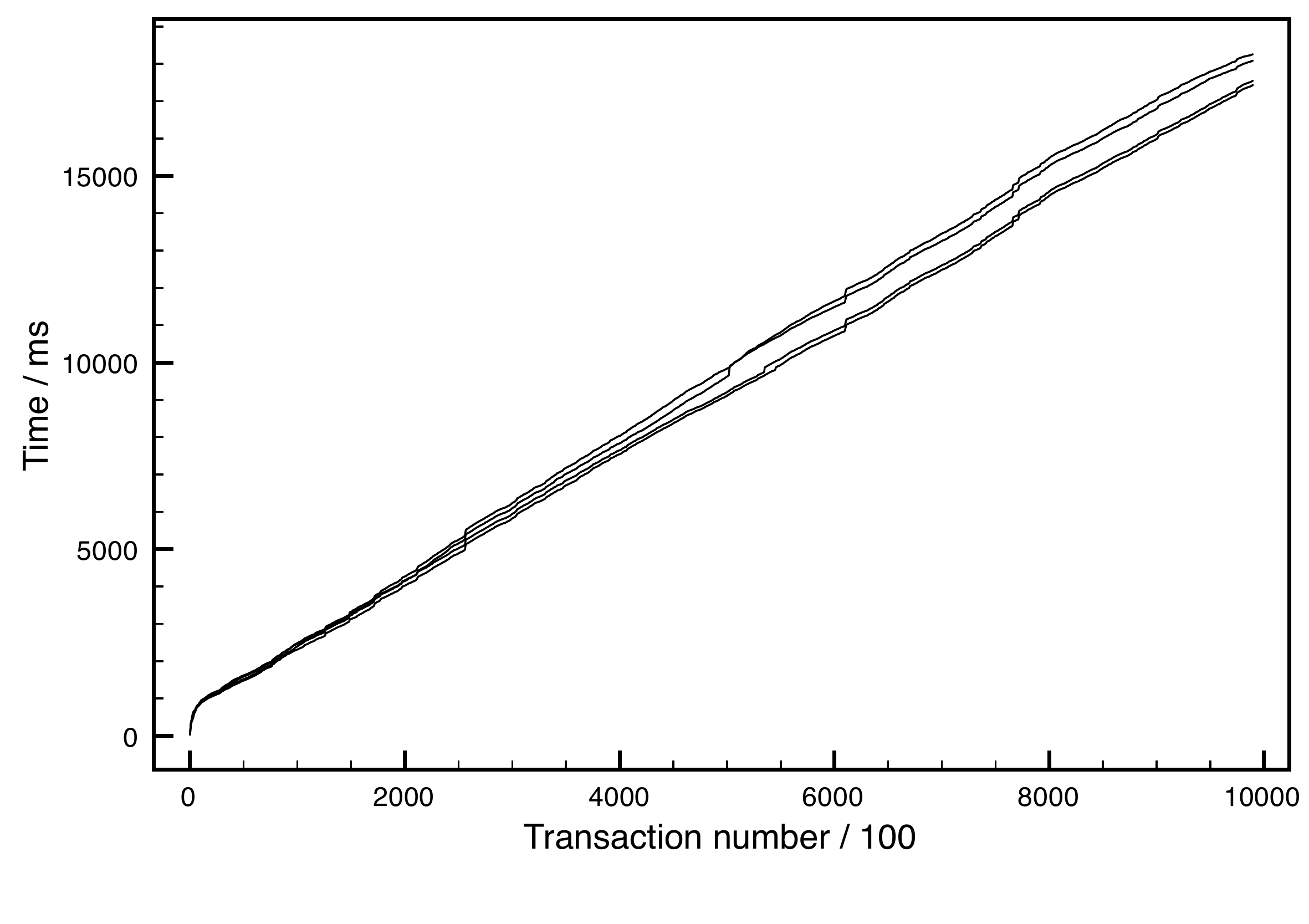}
\label{fig_kosarak-progress}
}
\vspace{-3mm}
\end{center}
\caption{\footnotesize  On the left side we see the progress of the current transaction number of each core running CRoP on the Nytimes data set, over a period of about 400 seconds and on the right-hand side the total time recorded for every 100th transaction when running CRoP with 4 cores on the Kosarak data set.}
\end{figure}

Somewhat surprisingly and inconsistent with the above results, we observed however that our experiments on a multicore CPU for various architectures do not always suggest very good scalability with increased number of cores. For example, the Webdocs dataset consists of long transactions only, thus the time for reading the input should not dominate. However, we observed 
that even if the processing time becomes better with more cores the advantages of having more cores become less and less pronounced. In particular, the running time between processing the data set with eight cores instead of four cores decreases by a factor of only 1.25 while the ratio of the running time between one core and two cores is about 1.85.  The reason is that each core has a small amount of dedicated L2 cache which is not sufficient to keep the whole interval of the hash table assigned to it and this leads to memory contention among processes. This means the potential of our method can be fully exploited when distributing the work to several processors or when we work with rather small sketches of the data stream. %For the exact running times for various datasets and more details on the experiments performed we refer to Appendix~\ref{app:perf}. 

More precisely, experiments have been carried out in order to verify how the algorithm scales, in terms of time,
when parallel computations are used.
We ran the algorithm on various datasets using different number of cores in order to highlight the parallel nature of the algorithm. Table~\ref{table:speedtest} reports some of the results we obtained.

As one can see the exact running times for the evaluation of the scalability of our algorithm appear somewhat cumbersome since the improvement in running times does not suggest very good scalability.  This is not due to a flaw in our implementation but a consequence of the specific Xeon E5570 architecture. Each core has a small amount of dedicated cache, namely 256 KB. In order to obtain correct estimates we need a large number of buckets which do not fit in cache and this leads to memory contention. For considerably smaller size of our sketch we achieve almost perfect scalability. We ran another set of experiments on a Mac Pro desktop equipped with Quad-Core Intel 2.8GHz. In this architecture there are 8 cores and a total of 24 MB cache available but 2 cores share 6 MB of cache. The running time for 4 parallel processes was about two times better than the running time for 2 parallel processes but the improvement for 8 processes was not as good, clearly indicating cache contention among 8 cores and an optimal use for 4 processes.

For a large dataset with a high number of pairs such as Nytimes our simple Java implementation nevertheless processed almost 100 million pairs per second on an 8-core Mac Pro, the total running time was about 110 seconds. The sketch contained 400000 buckets which was enough to find the exact counts of the top few hundred pairs. 
The throughput of a competing hash table solution can be upper bounded by the number of updates of random memory locations possible (disregarding time for hash function computation, and other overheads).
On the Mac Pro the number of such updates per second was estimated to around 50 millon per second, when updating a 1 GB table using 8 cores.
This means that we are at least a factor of 2 faster than any implementation based on a large, shared data structure.  

%More powerful CPU's with bigger dedicated cache memory for each core will be able to achieve good scalability for a larger number of buckets, thus the potential of our approach will be fully exploited in the future.   

%This means that the algorithm scales particularly well when the transaction are large.
%In those cases the parallelism of the computation speeds up considerably the algorithm.
% On datasets with smaller transactions, the system overhead for allocating resources
% for the processes seems to plays a role in running times.
\begin{table}
 \begin{center}
  \begin{tabular}{c||c|c|c} %dataset || number of processes or CPU || time used || time used by single process
    %  {}         & {} & \multicolumn{2}{|c|}{Python} & \multicolumn{2}{|c}{Java} \\
   {\bf Dataset} & {\bf \# of cores} &  {\bf \textit{ms} on \# cores} & {\bf \textit{ms} $1$ core} \\
   \hline\hline
   \multirow{3}{*}{Retail}      &       $8$        &  $1321$               & {} \\
     \cline{2-3}
     {}                         &       $4$        &  $1193$               & $2001$\\
     \cline{2-3}
       {}                       &       $2$        &  $1512$               & {}\\
    \hline
   \multirow{3}{*}{Kosarak} &       $8$        &  $1551$               & {} \\
     \cline{2-3}
      {}                        &       $4$        &  $1586$               & $2881$ \\
     \cline{2-3}
       {}                       &       $2$        &  $1997$               & {}\\
    \hline
   \multirow{3}{*}{Webdocs}     &       $8$        &  $299153$             & {} \\
     \cline{2-3}
      {}                        &       $4$        &  $357679$             & $891565$ \\
     \cline{2-3}
       {}                       &       $2$        &  $482111$             & {}\\
    \hline
   \multirow{3}{*}{Nytimes}     &       $8$        &  $443119$             & {} \\
     \cline{2-3}
      {}                        &       $4$        &  $524553$             & $1313698$ \\
     \cline{2-3}
       {}                       &       $2$        &  $689058$             & {}\\
    \hline
   \multirow{3}{*}{Wikipedia}   &       $8$        &  $27526403$           & {} \\
     \cline{2-3}
      {}                        &       $4$        &  $35397110$           & $93477243$ \\
     \cline{2-3}
       {}                       &       $2$        &  $53795313$              & {}\\
    \hline
  \end{tabular}
\caption{Experiments ran on an Intel Xeon E5570 2.93 GHz equipped with 23 GB of RAM;
the OS is GNU/Linux, kernel version 2.6.18.
The number of processes used is 8 for all four datasets.
Times are given in milliseconds (\textit{ms}). The number of buckets for the largest datasets is of the order $2^{20}$.
We can observe that in any case, millions of pairs per second were manipulated by
the algorithm.}\label{table:speedtest}
 \end{center}
\end{table}

\bibliographystyle{plain}
\bibliography{final_version}

\end{document}